\documentclass[publicdomain,creativecommons,sharealike,noncommercial]{%
eptcs}
\providecommand{\event}{GandALF 2012} 

\title{The discrete strategy improvement algorithm for parity games and complexity measures for
directed graphs}
\author{Felix Canavoi \qquad\quad Erich Gr\"{a}del \qquad\quad Roman Rabinovich
\institute{RWTH Aachen University\\ Germany}
\email{\{canavoi,graedel,rabinovich\}@logic.rwth-aachen.de}
}
\def\titlerunning{Strategy Improvement and Complexity Measures}
\def\authorrunning{F. Canavoi, E. Gr\"{a}del \& R. Rabinovich}

\usepackage[T1]{fontenc}
\usepackage[intlimits]{amsmath}
\usepackage{amsthm}
\usepackage{amsfonts}
\usepackage{amssymb}

\newtheorem{theorem}{Theorem}
\newtheorem{lemma}[theorem]{Lemma}

\newtheorem{remark}[theorem]{Remark}
\newtheorem{corollary}[theorem]{Corollary}

\usepackage{xspace}
\usepackage{textcomp}
\usepackage{float}
\usepackage{multirow}
\usepackage{txfonts} 
\usepackage{paralist}

\usepackage{tikz}
\usetikzlibrary{arrows,shapes,matrix,calc,%
chains,positioning,decorations.pathmorphing}

\newdimen\arrowsize
\newlength{\arrowlength}
\newlength{\arrowangle}
\newlength{\arrowthickness}

\pgfarrowsdeclare{slim}{slim}
{
\arrowsize=0.7pt
\advance\arrowsize by .5\pgflinewidth
\pgfarrowsleftextend{-4\arrowsize-.5\pgflinewidth}
\pgfarrowsrightextend{.5\pgflinewidth}
}
{
\advance\arrowsize by .5\pgflinewidth
\pgfsetdash{}{0pt} 
\pgfsetroundjoin 
\pgfsetroundcap 
\pgfpathmoveto{\pgfpointpolar{\arrowangle}{\arrowlength}}
\pgfpathlineto{\pgfpointorigin}
\pgfusepathqstroke
\pgfpathmoveto{\pgfpointpolar{\arrowangle}{\arrowlength}}
\pgfpathcurveto{\pgfpointpolar{\arrowangle*1.04}{
\arrowthickness+0.5*(\arrowlength-\arrowthickness)}}%
{\pgfpointpolar{\arrowangle*1.1}{
\arrowthickness+0.3*(\arrowlength-\arrowthickness)}}
{\pgfpointpolar{180}{\arrowthickness}}
\pgfpathlineto{\pgfpointorigin}
\pgfusepathqfill
\pgfpathmoveto{\pgfpointpolar{-\arrowangle}{\arrowlength}}
\pgfpathlineto{\pgfpointorigin}
\pgfusepathqstroke
\pgfpathmoveto{\pgfpointpolar{-\arrowangle}{\arrowlength}}
\pgfpathcurveto{\pgfpointpolar{-\arrowangle*1.04}{
\arrowthickness+0.5*(\arrowlength-\arrowthickness)}}%
{\pgfpointpolar{-\arrowangle*1.1}{
\arrowthickness+0.3*(\arrowlength-\arrowthickness)}}
{\pgfpointpolar{180}{\arrowthickness}}
\pgfpathlineto{\pgfpointorigin}
\pgfusepathqfill
}

\tikzstyle{vertex}=[draw,shape=circle, inner sep=0pt,minimum size=16]
\tikzstyle{dvertex}=[draw,dashed,shape=circle, inner sep=0pt,minimum size=16]
\tikzstyle{smallvertex}=[draw,shape=circle, inner sep=0pt,minimum size=8mm]
\tikzstyle{dashedvertex}=[draw,dashed,shape=circle, inner sep=0pt,%
minimum size=16]
\tikzstyle{smallcircle}=[circle,inner sep=1.5,fill=white, draw=black]
\tikzstyle{point}=[circle,inner sep=1,fill=black, draw=black]
\tikzstyle{path}=[-slim,thin,rounded corners] 
\tikzstyle{path1}=[-slim,thin,decorate,%
decoration={snake,amplitude=6pt,segment length=#1,%
pre length=#1/20,post length=#1/20}]
\tikzstyle{brace}=[thin,decorate,decoration=brace]
\tikzstyle{ie}=[thin,dashed,gray]

\tikzstyle{cop}=[star, star points=6, star point height=3.2,inner sep=3,minimum
size =2mm,semithick,fill=gray!80, draw=black]
\tikzstyle{ocop}=[star, star points=6, star point height=3.2,inner sep=3,minimum
size =2mm,semithick,fill=gray!40,draw=black,densely dotted]
\tikzstyle{robber}=[circle, inner sep=2,minimum size =2mm,semithick, draw=black]

\tikzstyle{slim loop}=[-slim,draw,to path={ .. controls +(60:1.5) and %
+(120:1.5) .. (\tikztotarget) \tikztonodes}]

\newcommand{\N}{\mathbb{N}}

\newcommand{\R}{\mathbb{R}}

\newcommand{\F}{\mathbb{F}}
\newcommand{\A}{\mathcal{A}}

\DeclareMathAlphabet{\mathsc}{OT1}{cmr}{m}{sc}

\newcommand{\ptime}{\ensuremath{\mathsc{Ptime}}\xspace}

\newcommand{\np}{\ensuremath{\mathsc{NP}}\xspace}

\newcommand{\0}{\emptyset}
\newcommand{\dunion}{\mathbin{\dot{\cup}}}
\newcommand{\done}{Done}

\newcommand{\tw}{treewidth\xspace}
\newcommand{\Tw}{Treewidth\xspace}
\DeclareMathOperator{\twg}{\mathrm{twG}}
\DeclareMathOperator{\twm}{tw}

\newcommand{\dagw}{DAG{}-width\xspace}

\DeclareMathOperator{\dagwg}{\mathrm{dagwG}}
\DeclareMathOperator{\dagwm}{dagw}

\newcommand{\cw}{cliquewidth\xspace}
\newcommand{\Cw}{Cliquewidth\xspace}
\DeclareMathOperator{\cwm}{cw}

\newcommand{\dpw}{directed pathwidth\xspace}
\newcommand{\Dpw}{Directed pathwidth\xspace}

\DeclareMathOperator{\dpwm}{dpw}

\newcommand{\kw}{Kelly{}-width\xspace}
\newcommand{\Kw}{Kelly{}-width\xspace}

\DeclareMathOperator{\kwg}{\mathrm{KwG}}

\DeclareMathOperator{\kwm}{Kw}

\newcommand{\ent}{entanglement\xspace}
\newcommand{\Ent}{Entanglement\xspace}
\DeclareMathOperator{\entg}{\mathrm{entG}}
\DeclareMathOperator{\entm}{ent}

\newcommand{\ie}{i.e.\xspace}
\newcommand{\eg}{e.g.\xspace}

\DeclareMathOperator{\Reach}{Reach}

\makeatletter
\newcommand{\@abbrev}[3]{
  \def\c@a@def##1{
      \if ##1.
        \relax
      \else
        \@ifdefinable{\@nameuse{#1##1}}{\@namedef{#1##1}{#2##1}}
        \expandafter\c@a@def
      \fi
    }
  \c@a@def #3.
}
\@abbrev{bb}{\mathbb}{ABCDEFGHIJKLMNOPQRSTUVWXYZ}
\@abbrev{bf}{\mathbf}{ABCDEFGHIJKLMNOPQRSTUVWXYZabcdefghijklmnopqrstuvwxyz}
\@abbrev{bit}{\boldsymbol}{ABCDEFGHIJKLMNOPQRSTUVWXYZabcdefghijklmnopqrstuvwxyz}
\@abbrev{cal}{\mathcal}{ABCDEFGHIJKLMNOPQRSTUVWXYZ}
\@abbrev{frak}{\mathfrak}{ABCDEFGHIJKLMNOPQRSTUVWXYZabcdefghijklmnopqrstuvwxyz}
\@abbrev{rm}{\mathrm}{ABCDEFGHIJKLMNOPQRSTUVWXYZabcdefghijklmnopqrstuvwxyz}
\@abbrev{scr}{\mathscr}{ABCDEFGHIJKLMNOPQRSTUVWXYZ}
\@abbrev{sf}{\mathsf}{ABCDEFGHIJKLMNOPQRSTUVWXYZabcdefghijklmnopqrstuvwxyz}
\@abbrev{ol}{\overline}{ABCDEFGHIJKLMNOPQRSTUVWXYZabcdefghijklmnopqrstuvwxyz}
\makeatother

\newcommand{\olcalG}{\overline{\calG}}

\begin{document}
\maketitle

\begin{abstract}
For some time the discrete strategy improvement algorithm due to
Jurdzi\'nski and V\"oge  had been considered as a candidate
for solving parity games in polynomial time. 
However, it has recently been proved by Oliver Friedmann that
the strategy improvement algorithm requires super-{}polynomially
many iteration steps, for all popular local improvements rules, 
including switch-all (also with Fearnley's snare memorisation), switch-best,
random-facet, random-edge, switch-half, least-{}recently-{}con\-si\-dered, and
Zadeh's Pivoting rule.

We analyse the examples provided by Friedmann in terms
of complexity measures for directed graphs such as \tw,
\dagw, \kw, \ent, \dpw, and \cw. 
It is known that for every class of parity games on which one of these
parameters is bounded,  the winning regions can be efficiently computed.
It turns out that with  respect to almost all of these measures, the complexity
of
Friedmann's  counterexamples is bounded, and indeed in most cases by very small
numbers. This analysis strengthens in some sense Friedmann's results and shows
that the discrete strategy improvement algorithm is even more limited than one
might have thought.   Not only does it require super-{}polynomial running time
in the
general case,  where the problem of polynomial-time solvability is open, it
even has super-{}polynomial lower time bounds on natural classes of parity games
on
which efficient algorithms are known.

\end{abstract}

\section{Introduction}

Parity games are a family of infinite two-player games on directed
graphs. They are important for several reasons. 
Many classes of games arising in practical applications admit
reductions to parity games (over larger game graphs). This is the
 case for games modelling reactive systems, with winning conditions
 specified in some temporal logic or in monadic second-order logic
 over infinite paths (S1S), for Muller games, but also for games with
 partial information appearing in the synthesis of distributed
 controllers. Further, parity games arise as the model-checking games for 
 \emph{fixed-point logics} such as the modal $\mu$-calculus or LFP, the
extension of
  first-order logic by least and greatest fixed-points. Conversely,
 winning regions of parity games (with a bounded number of priorities)
 are definable in both LFP and the $\mu$-calculus.
Parity games are also of crucial importance in the analysis of
structural  properties of fixed-point logics.

From an algorithmic point of view parity games are highly intriguing as well. 
It is an immediate consequence of the \emph{positional determinacy}
of parity games, that their winning regions
can be decided in NP $\cap$ Co-NP.
In fact, it was proved in \cite{Jurdzinski98}
that the problem is in 
$\textrm{UP} \cap \textrm{Co-UP}$, where UP denotes the class of NP-problems
with unique witnesses.
The best known deterministic algorithm 
has complexity $n^{O(\sqrt n)}$ \cite{JurdzinskiPatZwi06}.
For parity games with a number $d$ of 
priorities the progress measure lifting
algorithm by Jurdzi\'nski~\cite{Jurdzinski00} computes winning regions in time
$O(dm\cdot (2n/(d/2))^{d/2})=O(n^{d/2+O(1)})$, where 
$m$ is the number of edges, giving a
polynomial-time algorithm when $d$ is bounded.
The two approaches can be combined to achieve a worst-case
running time of $O(n^{d/3+O(1)})$ for solving parity games
with $d$ priorities, with $d=\sqrt n$ (see \cite[Chapter 3]{AptGra11}.

Although the question whether parity games are in
general solvable in \ptime  is still open,
there are efficient algorithms that solve parity games in special cases,
where the structural complexity of the underlying directed graphs,
measured by  numerical graph parameters, is low.
These include parity games of bounded \tw
\cite{Obdrzalek03}, bounded entanglement \cite{BerwangerGra05,BerwangerGraKaiRab12}, 
bounded \dagw \cite{BerwangerDawHunKreObd12},
bounded \kw \cite{HunterKre08}, or bounded \cw
\cite{Obdrzalek07}.

\medskip One algorithm that, for a long time, had been considered
as  a candidate for solving parity games in polynomial time
is the \emph{discrete strategy improvement algorithm} by Jurdzi\'{n}ski and
V\"{o}ge~\cite{JurdzinskiVoe00}. The basic idea behind the algorithm is to take
an
arbitrary
initial strategy for Player~$0$ and improve it step by step until an optimal
strategy is found. The algorithm is parametrized by an
\emph{improvement rule}. Indeed,  there are many possibilities
to improve the current strategy at any iteration step, and the
improvement rule determines the choice that is made.
Popular improvement rules are switch-all, switch-best, random-facet,
random-edge, switch-half and Zadeh's Pivoting rule. 
Although it is open whether there is an
improvement rule that results in a polynomial worst-case runtime of the
strategy improvement algorithm, Friedmann~\cite{FriedmannPhD} was
able to show
that there are super-{}polynomial lower bounds for all popular 
improvement rules mentioned above. For each of these rules, Friedmann 
constructed a family of parity games
on which the strategy improvement algorithm requires super-{}polynomial
running time.

In this paper we analyse the examples provided by Friedmann in terms
of complexity measures for directed graphs. It turns out that with 
respect to most of these measures, the complexity of Friedmann's counterexamples
is bounded, and indeed in most cases by very small numbers.
This analysis strengthens in some sense Friedmann's results and shows that the
discrete strategy improvement algorithm is even more limited than one might
have thought.  
Not only does it require super-{}polynomial running time in the general case, 
where the problem of polynomial-time solvability is open, it
even has super-{}polynomial lower time bounds on natural classes of parity games
on
which efficient algorithms are known.

\section{The strategy improvement algorithm}\label{sec_strat_improv_alg}

We assume that the reader is familiar with basic notions and terminology
on parity games. We shall now briefly discuss the discrete strategy improvement
algorithm and the different improvement rules that parametrize it.
For the purpose of this paper a precise understanding of the
algorithm is not needed.
The idea of the strategy improvement algorithm is that one can
compute an optimal strategy of a player by starting with an arbitrary intitial
strategy
and improve it step by step, depending on a discrete valuation of plays and
strategies, and on a rule that governs the choices of local changes
(switches) of the current strategy.  

It is well-known that parity games are determined by positional strategies,
\ie strategies which at each position just select one of the outgoing edges,
independent of the history of a play.
The discrete valuation defined by Jurdzi\'nski and V\"oge \cite{JurdzinskiVoe00}
measures how good a play is for Player~0 in a more refined way than
just winning or losing. Given a current strategy one can then,
at each position of Player~0, consider the possible  local changes,
\ie the switches of the outgoing edges, and select a
locally best posibility.
Rules that describe how to combine such switches in one improvement step are
called \emph{improvement rules} and parametrise the algorithm.

The \emph{switch-all} or \emph{locally optimizing} rule~\cite{JurdzinskiVoe00}
regards each vertex independently and performs the best possible switch for
every vertex. In other words, for every vertex, it computes the best
improvement of the strategy \emph{at that vertex} assuming that the strategy
remains
unchanged at other vertices. However, the switch is done simultaneously at each
vertex. The \emph{switch-{}best} or the \emph{globally optimising
rule} takes cross-effects of improving switches into account and
applies in every iteration step a best possible combination of switches. 

The \emph{random-edge} rule applies a single improving switch at some vertex
chosen randomly and the improvement rule \emph{switch-half}
applies an improving switch at every vertex with probability~$1/2$. The 
\emph{random-{}facet} rule chooses randomly an edge~$e$ 
leaving a Player~0 vertex and computes recursively a winning strategy
$\sigma^*$ 
on the graph without~$e$. If taking~$e$ is not an improvement, $\sigma^*$ is 
optimal, otherwise~$\sigma^*$ switched to~$e$ is the new initial strategy. The 
\emph{least-entered} rule switches at a vertex at that the least 
number of switches has been performed so far. Cunningham's
\emph{least-recently-considered} or \emph{round-robin} rule 
fixes an initial ordering on all Player~0 vertices first, and then selects the
next vertex to switch at in a round-robin manner. Fearnley introduced
in~\cite{Fearnley10} \emph{snare memorisation}. It can be seen as an extension
of a basic improvement rule by a snare rule that memorises certain structures of
a game to avoid reoccurring patterns.

All local improvement rules discussed here can be computed in polynomial
time~\cite{JurdzinskiVoe00,Schewe08}. Hence the 
running time of the algorithm on a game depends primarily on the number of
improvement steps. 
In a series of papers and in his dissertation, Friedmann has constructed, 
for each of the above mentioned improvement rules, a class of parity games 
on which  the strategy improvement algorithm requires 
super-{}polynomially many iteration steps, with respect  to the size of the 
game~\cite{FriedmannPhD,Friedmann12a,Friedmann12b}.
We shall analyse these games in terms of certain complexity measures for
directed
graphs which we describe in the next section.

\section{Complexity measures for directed graphs}

For most of the complexity measures we shall work with a 
characterisation in terms of so-called \emph{graph searching games}, that allow
us a more 
intuitive point of view on
the measures and give us an easier way to analyse the graphs in question. A
graph searching game is played on a graph by a team of cops and a robber. In any
position, the robber is on a vertex of the graph and each cop either also
occupies a vertex or is outside of the graph. The robber can move between
vertices along cop-free paths in the graph, \ie  paths whose vertices are not
occupied by cops. The moves of the cops have typically no restrictions. The aim
of
the cops is to capture the robber, \ie to force him in a position where he has
no legal moves. Precise rules of moves characterise a complexity measure of the
graph. The value of the measure is the minimal number of cops needed to capture
the robber (${}-1$ in some cases). Hence, on simple graphs,  few cops suffice
to win wheras complex graphs demand many cops to capture the robber.

In that way, \tw, \dagw, \kw, \dpw and \ent can be described. Another measure
that we shall consider is \cw, for which no characterisation by games is
known. Recall that common definitions of such measures are usually given
by means of appropriate decompositions of the graph into small parts that are
connected in a simple way: as
a directed path for \dpw, as a tree for \tw, as a DAG for \dagw and \kw, or as a
parse tree for \cw. The maximal
size of a part in a decomposition corresponds
to the minimal number of cops needed to capture the robber on the 
graph (except for \ent, for which no corresponding decomposition is known). Such
decompositions can be used to provide efficient algorithms for problems
that are difficult (\eg \np-complete) in general, on graph classes where the
values 
of the respective measure is bounded. 
In particular, this is the case for parity games. In a series of papers it has
been
shown that parity games can be solved in \ptime on graph classes with bounded
\tw, \dpw, \dagw, \kw, \ent or \cw. In the following, we define all
complexity measures discussed above by their characterisations in terms of
graph
searching games except \cw, for which we give an inductive definition. 

\paragraph{\Tw.}

\Tw is a classical measure of cyclicity on undirected
graphs. It measures how close a graph is to being a tree. The \tw
game $\twg_k(\calG)$ is played on an undirected graph $\calG=(V,E)$ by a team
of~$k$ cops and a robber, whereby~$k$ is a parameter of the game. Initially,
there are no cops on the graph and the robber chooses an arbitrary vertex and
occupies it in the first move. The players move alternating. A cop
position is a tuple~$(C,v)$ where $C\subseteq V$ with $|C|\le k$ is the set of
vertices occupied by cops (if $k > |C|$, the remaining cops are considered to
be outside of~$\calG$) and~$v\notin C$ is the vertex occupied by the robber. The
cops can move to a position $(C,C',v)$ with $|C'| \le k$. Intuitively, they
announce their next placement~$C'$ and take cops from $C'\setminus C$ away
from~$\calG$. The robber positions are of the form $(C,C',v)$. The robber can
run along paths on the graph whose vertices are not occupied by cops, \ie the
next (cop) position may be $(C',w)$ where $w\in\Reach_{\calG-(C\cap
C')}(v)\setminus C'$, \ie~$w$ is reachable from~$v$ in~$\calG-(C\cap
C')$. Thus the cops are placed on the vertices they announced in
their previous move; furthermore, only those cops prevent the robber to run who
are both in the previous and in the next placements. However, the robber is not
permitted to go to a vertex which will be occupied by a cop in the next
position.

The robber is captured in a position $(C,C',v)$ if he has no legal move: all
neighbours of~$v$ are in $C\cap C'$ and a cop is about to occupy his vertex, \ie
$v\in C'$. A play is monotone if the robber can never reach a vertex that has
already been unavailable for him. It suffices to demand that in any move, the
robber is not able to reach a vertex that has just been left by a cop. Formally,
a play is monotone if, for every cop move $(C,v)\to(C,C',v)$ in the play, we
have $\Reach_{\calG-(C\cap C')}(v) \cap (C\setminus C') = \0$.

The cops win a monotone play if it ends in a position in that the robber is
captured. Infinite or non-{}monotone plays are won by the robber. A (positional)
strategy for the cops is a function $\sigma\colon 2^V\times V \to 2^V$ which
prescribes, for every cop position $(C,v)$, the next placement $\sigma(C,v)$.
Similarly, a (positional) strategy for the robber is a function $\rho\colon
2^v\times 2^V \to V$ which maps a robber position $(C,C',v)$ to a cop
position $(C',w)$ with $w\in \Reach_{\calG - (C\cap C')}(v)$. A play $\pi$ is
consistent with a strategy $\sigma$ (or $\rho$) if every move in the play is
made according to $\sigma$ (or $\rho$). A strategy for a player is winning if
he wins every play consistent with that strategy. As the winning conditions for
both players are Boolean combinations of reachability and safety, it suffices to
consider only positional strategies.

The minimal number~$k$ such that the cops have a winning strategy in
$\twg_{k+1}(\calG)$ is the \emph{\tw} $\twm(\calG)$ of~$\calG$. If $\calG =
(V,E)$
is a directed graph then $\twm(\calG)$ is $\twm(\overline{\calG})$ where
$\overline{\calG} =
(V,\overline{E})$ and~$\overline{E}$ is the symmetric closure of~$E$.

\paragraph{\dagw.}

The \dagw game~$\dagwg_k(\calG)$~\cite{BerwangerDawHunKreObd12} is played on a
directed graph~$\calG$ in the
same way as the \tw game, but the edge relation of the graph is not
symmetrised. Note that the meaning of the
reachability relation $\Reach{}$ on directed graphs is, of course, different
from the reachability relation on undirected graphs. In a \dagw game, the robber
is
allowed to run only along \emph{directed paths}. The \dagw $\dagwm(\calG)$ of a
graph $\calG$ is the minimal number $k$ such that the cops have a winning
strategy in the game $\dagwg_k(\calG)$. Note the difference to \tw where the
parameter of the game is defined by $k+1$ in order to make forests have \tw~$1$.

\paragraph{\Kw.}

The \kw game~$\kwg_k(\calG)$ is played on a directed graph~$\calG$ in the same
way as the \dagw game, but the robber is, first, invisible for the cops and,
second, inert~\cite{HunterKre08}. Invisibility means that a winning strategy for
the cops must not depend on the robber vertex and the cops can make
assumptions about it only from their own moves. Inertness of the robber means
that the robber can change his vertex only if a cop has announced to occupy the
robber vertex in the next position. Formally, a cop position is a tuple
$(C,R)$ where $C$ is as before and $R\subseteq V$ is disjoint with~$C$. The
cops can move to a
robber position $(C,C',R)$. The moves of the robber are determined by the
current position, so, in
fact, we have a one{}-player game: the next position is $(C',R')$ where $R' =
(R\cup \Reach_{\calG- (C\cap C')}(R\cap C'))\setminus C'$. The term
$\Reach_{\dots}(R\cap C')$ describes the inertness of the robber and the term
$R\cup \dots$ means that the robber may still be on a previous vertex if no
cop is about to occupy it. \Kw is defined analogously to \dagw.

\paragraph{\Dpw.}

The \dpw game is played as the \kw game, but the robber is not inert. Formally,
the position
following a robber position $(C,C',R)$ is $(C',R')$ where $R' = \Reach_{\calG-
(C\cap C')}(R\cap C')\setminus C'$. Similar to \tw, \dpw $\dpwm(\calG)$
of~$\calG$ is the
minimal number~$k$ such that the cops have a winning strategy in
$\dpwm_{k+1}(\calG)$.

\paragraph{\Ent.}

The \ent game $\entg_k(\calG)$~\cite{BerwangerGra05} is slightly different from
the games defined 
above. First, the robber can move only along an edge rather than along a whole
path.
Second, he is \emph{obliged} to leave his vertex, no matter if a cop is about to
occupy it or not (thus no cops are needed on an acyclic graph). Third, the cops
are
restricted in their moves as well. In a cop position $(C,v)$, one cop can
go to the~$v$, other cops must remain on their vertices. Another possibility for
the cops is to stay idle.
More formally, cop
positions are of the form $(C,v)$ and the cops can move to some
position~$(C',v)$ where $C'=C$, or $C' = C\cup \{v\}$ (if a new cop comes in to
the
graph), or $C' = (C\cup \{v\}) \setminus \{w\}$ where $w\in C$ is distinct
from~$v$. From a position $(C',v)$, the robber can move to a position~$(C',v')$
where $(v,v')\in E$ and~$v'\notin C'$. Unlike all games above, in the \ent game,
the cops do not need to play monotonically, so they win all finite plays and the
robber wins all infinite plays. \Ent $\entm(\calG)$ of a graph $\calG$ is the
minimal number~$k$ such that the cops have a winning strategy
in~$\entg_k(\calG)$.

\paragraph{\Cw} was introduced in~\cite{CourcelleEngRoz93}. Let $C$ be a finite
set of labels. A $C$-labelled graph is a tuple~$\calG = (V,E,\gamma)$ where
$\gamma \colon V \to C$ is a map that labels the vertices of~$\calG$ with
colours from~$C$. An $a$-port is a vertex with colour~$a$. Let~$k$ be a positive
natural number and let~$|C|\le k$. The class~$\calC_k$ of graphs of \cw at
most~$k$ is defined
inductively by the following operations.

\begin{enumerate}[(1)]
\item For every $a \in C$, a single $a$-port without edges is in~$\calC_k$.

\item If $\calG_1 = (V_1,E_1,\gamma_1)$ and $\calG_2 = (V_2,E_2,\gamma_2)$ are
in $\calC_k$ then the disjoint union $\calG_1 \oplus \calG_2 = (V,E,\gamma)$ of
$\calG_1$ and $\calG_2$ is in~$\calC_k$ where $V = V_1 \dunion V_2$, $E=
E_1\dunion E_2$, and $\gamma(v) = \gamma_1(v)$ if $v\in V_1$ and $\gamma(v) =
\gamma_2(v)$ if $v\in V_2$.

\item If $\calG = (V,E,\gamma)$ is in $\calC_k$ then the  graph $\calG'$
obtained by recolouring every $a$-port to a $b$-port is in $\calC_k$, \ie
$\calG' =
(V,E,\gamma')$ where $\gamma'(v) = \gamma'(v)$ if $\gamma(v)\neq a$ and
$\gamma'(v) = b$ otherwise.

\item If $\calG = (V,E,\gamma)$ is in~$\calC_k$ then the graph $\calG'$ obtained
by connecting all $a$-ports to all $b$-ports is in~$\calC_k$, \ie $\calG' =
(V,E',\gamma)$ where $E' = E \cup \{(v,w) \mid \gamma(v) = a \text{ and }
\gamma(w) = b\}$.
\end{enumerate}
The \cw $\cwm(\calG)$ of a graph $\calG = (V,E,\gamma)$ is the least~$k$ such
that the graph $(V,E)$ is in $\calC_k$.

The following theorem is a combination of results proved in
\cite{BerwangerDawHunKreObd12,BerwangerGra05,BerwangerGraKaiRab12,HunterKre08,Obdrzalek03,Obdrzalek07}.

\begin{theorem}\label{thm_parity_fast}
Let $C$ be any class of finite graphs on which
at least one of the following measures is bounded: \tw, \dpw, \dagw, \cw, \kw,
\ent.
Then the winning regions of parity games on $C$ are computable in polynomial
time.
\end{theorem}

It follows directly from the definitions that \dagw and \kw are bounded in \dpw.

\begin{theorem}\label{thm_relate_dagw_kw_pw}
For a graph $\calG$, we have $\dagwm(\calG) \le \dpwm(\calG)+1$ and 
$\kwm(\calG) \le \dpwm(\calG)+1$.
\end{theorem}

\section{Friedmann's counterexamples}

We now describe and analyze the graphs that underlie Friedmann's counterexample
games.
For most of the rules, these graphs have a rather similar structure, and for
reasons of
space we give a detailed pesentation just for the examples for the switch-all
rule
and Zadeh's Pivoting rule. Our analysis of the examples for the other rules
will be summarized in a table, proofs will be given in the full version of this
paper.

\subsection{The Switch-All Rule}
 
For $n \in \N \setminus \{0\}$, the graph $\calG_n=(V_n,E_{n})$  underlying
Friedmann's games against the switch-all rule can defined as follows.
The set of vertices is 
\[V_n\coloneqq \{x,s,c,r\} \cup \{t_{i},a_{i}: 1\leq i \leq 2n\} \cup
\{d_{i},e_{i},g_{i},k_{i},f_{i},h_{i}:1\leq i\leq n\}.\] 
The set of edges and the graph $\calG_3$ are given in
Figure~\ref{fig_switch_all}. The graph~$\calG_n$
consists of cycle gadgets induced by $\{d_i,e_i\}$ each encoding a bit which is
considered to be set if the current strategy of Player~0 is to move from~$d_i$
to~$e_i$ and unset otherwise.  Intuitively, the strategy improvement
algorithm with the switch-all rule starts from the state where all bits are
unset and increases the bit
counter by one in each round. The subgraph induced by all~$h_j$, $k_j$, $g_j$,
and~$f_j$, for~$j\le n$ guarantees the algorithm to swap the least significant
bit and the
subgraph induced by~$a_j$ and $t_j$, for~$j\le 2n$ ensures that the other bits
to change are swapped as well, see~\cite{FriedmannPhD} for details. 

\begin{figure}[h]
\begin{minipage}{0.7\linewidth}

\begin{tikzpicture}[scale=0.7]

\draw[use as bounding box,opacity=0] (0,-2) rectangle (13,16.5);


  \path (0,1) node[vertex] (t6) {$t_{6}$};

  \path (2.2,1) node[vertex] (t5) {$t_{5}$};

  \path (4.4,1) node[vertex] (t4) {$t_{4}$};

  \path (6.6,1) node[vertex] (t3){$t_{3}$};

  \path (8.8,1) node[vertex] (t2){$t_{2}$};

  \path (11,1) node[vertex] (t1) {$t_{1}$};

  \path (13,1) node[vertex] (c) {$c$};
  \node[dashedvertex] (r) at (13,2.5){r};
  \node[dashedvertex] (s) at (13,-0.5){s};
  \draw[-slim] (c)to (r);
  \draw[-slim] (c)to (s);
  

  \node[dashedvertex] (r6) at (-0.48,-0.5) {$r$}; \draw [-slim](t6) to (r6);

  \node[dashedvertex] (s6) at (0.48,-0.5) {$s$};  \draw [-slim](t6) to (s6);

  \node[dashedvertex] (r5) at (1.72,-0.5) {$r$}; \draw [-slim](t5) to (r5);

  \node[dashedvertex] (s5) at (2.68,-0.5) {$s$};  \draw [-slim](t5) to (s5);

  \node[dashedvertex] (r4) at (3.92,-0.5) {$r$}; \draw [-slim](t4) to (r4);

  \node[dashedvertex] (s4) at (4.88,-0.5) {$s$};  \draw [-slim](t4) to (s4);

  \node[dashedvertex] (r3) at (6.12,-0.5) {$r$}; \draw [-slim](t3) to (r3);

  \node[dashedvertex] (s3) at (7.08,-0.5) {$s$};  \draw [-slim](t3) to (s3);

  \node[dashedvertex] (r2) at (8.32,-0.5) {$r$}; \draw [-slim](t2) to (r2);

  \node[dashedvertex] (s2) at (9.28,-0.5) {$s$};  \draw [-slim](t2) to (s2);

  \node[dashedvertex] (r1) at (10.52,-0.5) {$r$}; \draw [-slim](t1) to (r1);

  \node[dashedvertex] (s1) at (11.48,-0.5) {$s$};  \draw [-slim](t1) to (s1);


  \path (0,2.5) node[vertex] (a6) {$a_{6}$};

  \path (2.2,2.5) node[vertex] (a5) {$a_{5}$};

  \path (4.4,2.5) node[vertex] (a4) {$a_{4}$};

  \path (6.6,2.5) node[vertex] (a3) {$a_{3}$};

  \path (8.8,2.5) node[vertex] (a2) {$a_{2}$};

  \path (11,2.5) node[vertex] (a1) {$a_{1}$};

  \draw [-slim](t6) to (t5);\draw [-slim](t5) to (t4);\draw [-slim](t4) to (t3);
  \draw [-slim](t3) to (t2);\draw [-slim](t2) to (t1);\draw [-slim](t1) to (c);
  \draw [-slim](a6) to (t6);\draw [-slim](a5) to (t5);\draw [-slim](a4) to (t4);
  \draw [-slim](a3) to (t3);\draw [-slim](a2) to (t2);\draw [-slim](a1) to (t1);



  \path (1.1,6) node[vertex] (d3) {$d_{3}$};

  \path (5.5,6) node[vertex] (d2) {$d_{2}$};

  \path (10,6) node[vertex] (d1) {$d_{1}$};

  \draw [-slim](d3) to (a6);\draw [-slim](d3) to (a5);\draw [-slim](d3) to
(a4);\draw
       [-slim](d3) to (a3);\draw [-slim](d3) to (a2);\draw [-slim](d3) to (a1);

  \draw [-slim](d2) to (a4);\draw [-slim](d2) to (a3);\draw [-slim](d2) to
(a2);\draw
       [-slim](d2) to (a1);

  \draw [-slim](d1) to (a2);\draw [-slim](d1) to (a1);


  \node[dashedvertex] (r7) at (-0.1,6.60) {$r$}; \draw [-slim](d3) to (r7);
 
  \node[dashedvertex] (s7) at (-0.1,5.6) {$s$};  \draw [-slim](d3) to (s7);

  \node[dashedvertex] (r8) at (4.3,6.60) {$r$}; \draw [-slim](d2) to (r8);

  \node[dashedvertex] (s8) at (4.3,5.6) {$s$};  \draw [-slim](d2) to (s8);

  \node[dashedvertex] (r9) at (8.8,6.60) {$r$}; \draw [-slim](d1) to (r9);

  \node[dashedvertex] (s9) at (8.8,5.6) {$s$};  \draw [-slim](d1) to (s9);



  \path (1.1,7.5) node[vertex] (e3) {$e_{3}$};

  \path (5.5,7.5) node[vertex] (e2)  {$e_{2}$};

  \path (10,7.5) node[vertex] (e1) {$e_{1}$};

  \draw [-slim](d3)to[bend right=30](e3);\draw [-slim](e3)to[bend right=30](d3);
  \draw [-slim](d2)to[bend right=30](e2);\draw [-slim](e2)to[bend right=30](d2);
  \draw [-slim](d1)to[bend right=30](e1);\draw [-slim](e1)to[bend right=30](d1);

  

  \path (3.1,8) node[vertex] (f3) {$f_{3}$};
  \path (7.5,8) node[vertex] (f2)  {$f_{2}$};
  \path (12,8) node[vertex] (f1) {$f_{1}$};

  \draw [-slim](f3) to (e3);\draw [-slim](f2) to (e2);\draw [-slim](f1) to (e1);



  \path (3.1,11) node[vertex] (g3) {$g_{3}$};
  \path (7.5,11) node[vertex] (g2) {$g_{2}$};
  \path (12,11) node[vertex] (g1) {$g_{1}$};
  
  \draw [-slim](g3) to (f3);\draw [-slim](g2) to (f2);\draw [-slim](g1) to (f1);



  \path (1.1,11) node[vertex] (h3) {$h_{3}$};
  \path (5.5,11) node[vertex] (h2) {$h_{2}$};
  \path (10,11) node[vertex] (h1) {$h_{1}$};

  \draw [-slim](e3) to (h3);\draw [-slim](e2) to (h2);\draw [-slim](e1) to (h1);


  
  \path (1.1,13) node[vertex] (k3) {$k_{3}$};
  \path (5.5,13) node[vertex] (k2) {$k_{2}$};
  \path (10,13) node[vertex] (k1) {$k_{1}$};

  \draw [-slim](h3) to (k3);\draw [-slim](h2) to (k2);\draw [-slim](h1) to (k1);
  \draw [-slim](g3) to (k3);\draw [-slim](g2) to (k2);\draw [-slim](g1) to (k1);
  \draw [-slim](k1) to (g3);\draw [-slim](k1) to (g2);
  \draw [-slim](k2) to (g3);



  \node[vertex] at (5.5,15) (x) {$x$};  
  \draw (x) to [slim loop] (x);

  \draw [-slim](k3) to (x);
  \draw [-slim](k2) to (x);
  \draw [-slim](k1) to (x);

  \path (13,9) node[vertex] (r) {$r$};
  \path (0,10) node[vertex] (s) {$s$};

  \draw [-slim,bend left=10](r) to (g3);
  \draw [-slim](r) to (g2);
  \draw [-slim](r) to (g1);
  \draw [-slim](s) to (f3);
  \draw [-slim](s) to (f2);
  \draw [-slim](s) to (f1);
  
  \draw [-slim](s)to[bend left=70](x);
  \draw [-slim](r)to[bend right=60](x);


\end{tikzpicture}

\end{minipage}
\begin{minipage}{0.2\linewidth}

\begin{tabular}{|c|c|}
\hline
Node & Successors \\
\noalign{\hrule height 0.45mm}
$t_{1}$		& $\{ s,r,c \}$ \\
\hline
$t_{i>1}$ 	& $\{ s,r,t_{i-1} \}$ \\
\hline
$a_{i}$		& $\{ t_{i} \}$ \\
\hline 
$c$ 		& $\{ s,r \}$ \\
\hline \hline
\multirow{2}{*}{$d_{i}$} & $\{ s,r,e_{i} \}$\\
		& ${} \cup \{a_{j}|j \le 2i\}$ \\
\hline
$e_{i}$ 	& $\{ d_{i}, h_{i} \}$ \\
\hline
$g_{i}$ 	& $\{ f_{i}, k_{i} \}$ \\
\hline
\multirow{2}{*}{$k_{i}$} & $\{x\}$\\
		& ${}\cup \{g_{j}|i<j\leq n \}$ \\
\hline
$f_{i}$ 	& $\{ e_{i} \}$ \\
\hline
$h_{i}$ 	& $\{ k_{i} \}$ \\
\hline \hline
\multirow{2}{*}{$s$} & $\{x\}$\\
		& ${}\cup \{f_{j}|j\leq n \}$ \\
\hline
\multirow{2}{*}{$r$} & $\{x\}$\\
		& ${}\cup \{g_{j}|j\leq n \}$ \\
\hline
$x$		& $\{x\}$ \\
\hline
\end{tabular}

\end{minipage}

\caption{The graph~$\calG_3$ and the edge relation of~$\calG_n$ for
the counterexample to the switch-all rule.}
\label{fig_switch_all}

\end{figure}
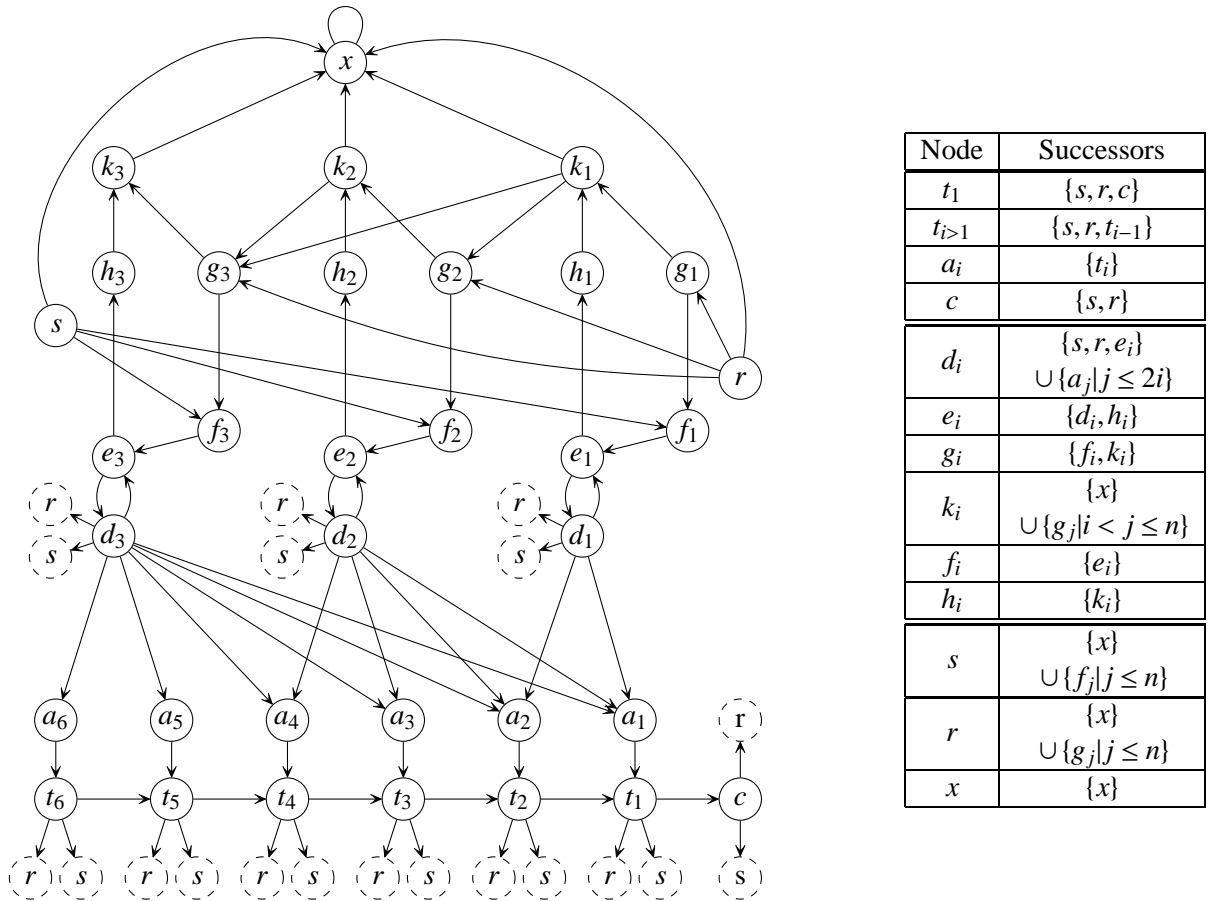

Friedmann showed in~\cite{FriedmannPhD}
that, for every $n>0$, there is a parity game of size $O(n^2)$ with underlying
graph~$\calG_n$ such that the strategy improvement algorithm with 
the switch-all rule requires at least~$2^n$ improvement steps on that game.

We shall now establish upper bounds for \dagw, \kw, \dpw, \ent, and \cw of 
the graphs~$\calG_n$, which imply by Theorem~\ref{thm_parity_fast}, 
that Friedmann's games belong to natural classes of parity games 
that can be solved efficiently by other approaches than the strategy improvement
algorithm.
We start with an analysis of \tw of~$\calG_n$ and show that it is unbounded 
on the class of graphs~$\calG_n$. Recall that \tw of a directed 
graph~$\calG = (V,E)$ is defined by the \tw of~$\olcalG=(V,\olE)$ where $\olE$
is 
the symmetrical closure of~$E$.
The reason for \tw to be unbounded is that it contains arbitrarily large
complete bipartite graphs $\calK_{n,n}$ as
subgraphs, whereby~$\twm(\calK_{n,n}) = n$. Indeed, every vertex has $n$ direct 
successors, so if the robber is caught staying on a vertex~$v$, all successors
of~$v$
and~$v$ itself must be occupied by cops.
The following lemma shows that we can find an arbitrary complete bipartite graph
as a subgraph
of a graph of the family $\olcalG_{n}$.

\begin{lemma}\label{lemma_bipartite_in_G_n}
For every $k>0$, there is some~$n>0$ such that~$\olcalG_n$ has~$\calK_{k,k}$
as a subgraph.
\end{lemma}
\begin{proof}
Choose $n\coloneqq \lceil \frac{k}{2} \rceil + k - 1$. The vertex $d_{\lceil
\frac{k}{2}\rceil}$ is the first of the vertices $d_{1},\dots,d_{n}$ to be
connected to the vertices $A\coloneqq \{a_{j}:j\leq k\}$.
The~$k-1$ vertices $d_{i}$, $i=\lceil \frac{k}{2} \rceil + 1,\dots, \lceil
\frac{k}{2} \rceil + k - 1$ are connected to each vertex of~$A$ as well. Neither
the vertices of~$A$ are directly connected to one another, nor are the vertices
of $B\coloneqq \{d_{i} \mid i=\lceil \frac{k}{2} \rceil,\dots, \lceil
\frac{k}{2}
\rceil + k - 1\}$. It follows that $\olcalG[A\cup B]$ is isomorphic
to~$\calK_{k,k}$.
\end{proof}

\begin{corollary}\label{cor_tw_unbounded_on_switch_all}
For every $k>0$, there is some $n>0$ such that $\twm(\calG_n) > k$.
\end{corollary}

\begin{remark}
Although the treewidth of~$\calG_n$ is unbounded, there is another
class of graphs with bounded treewidth, such that the strategy improvement
algorithm with switch-all rule requires super-polynomial time. 
We shall see in Section~\ref{subsec_other_rules} that for the random-edge
rule, Friedmann's counterexample class has bounded \tw. In fact, that class
requires super-polynomial time also for the switch-all
rule, see~\cite{FriedmannPhD} for details.
\end{remark}

Now we prove that \dpw of graphs~$\calG_n$ is bounded, which leads to
boundedness of \dagw and \kw.

\begin{theorem}\label{thm_dpw_bounded}
For all~$n>0$, we have $\dpwm(\calG_n) \le 3$. 
\end{theorem}
\begin{proof}
We describe a monotone winning strategy for~$4$ cops in the \dpw game.
First,~$r$ and~$s$
are occupied by two cops who will stay there until the robber is caught. In the
next round,
the two other cops expel the robber from all vertices~$d_i$, $e_i$, $f_i$,
$g_i$, $h_i$, and $k_i$
(if he is there).
For $i = 1,\dots,n$, starting with $i=1$ the cops place a cop on~$e_i$ and then 
visit with the last remaining cop vertices $d_i$, $f_i$, $g_i$, $h_i$, and $k_i$
in that order.

The robber may be on vertex~$x$, or in the part of the graph induced by~$a_i$,
$t_i$ and~$c$, for~$i\in\{1,\dots,2n\}$. The cop from~$k_n$ (one of those not
on~$r$ or $s$) visits $x$ and then
$a_n$, $t_n$, $a_{n-1}$, $t_{n-1}$, \dots, $a_1$, $t_1$ in that order and
finally~$c$.
Obviously, the described strategy for~$4$ cops is monotone and guarantees that
the robber
is captured.
\end{proof}

By Theorem~\ref{thm_relate_dagw_kw_pw}, we get the following corollary.

\begin{corollary}\label{cor_dagw_and_kw_unbounded_on_switch_all}
For all $n>0$, we have $\dagwm(\calG_n) \le 4$ and $\kwm(\calG_n) \le 4$.
\end{corollary}

We modify the strategy from the proof of Theorem~\ref{thm_dpw_bounded} to obtain
a winning strategy for~$3$ cops in the \ent game. That is necessary, as in the
latter,
the cops are not permitted to be placed on a vertex which is not occupied by the
robber.
We first need a lemma from~\cite{BerwangerGra05}.

\begin{lemma}\label{lemma_ent_one}
The entanglement of a graph is one if, and only if, it is not acyclic and each of its
strongly connected components contains a vertex whose removal makes the component acyclic.
\end{lemma}

\begin{theorem}\label{thm_ent_bounded}
For all $n > 0$, we have $\entm(\calG_n) \leq 3$.
\end{theorem}
\begin{proof}
\newcommand{\G}{\ensuremath{\calG_n^{r,s}}}
Let $\G$ be the graph which is obtained from~$\calG_n$ by deleting vertices~$r$ 
and~$s$ and all adjacent edges, \ie $\G = \calG_n[V_n\setminus\{r,s\}]$. 
The only strongly connected components of $\G$ where the robber can remain are
the one induced by~$x$ and those induced by $\{d_i,e_i\}$. All other components
are singletons without self-loops, so the robber can stay there only for one
move.
Each of the components induced by $x$ or by $\{d_i,e_i\}$ have a vertex whose
removal 
makes the component acyclic. By Lemma~\ref{lemma_ent_one}, there is a 
strategy~$\sigma$ for one cop to catch the robber on~$\G$. Thus it suffices to
prove 
that the cops can occupy~$r$ and~$s$. They use one cop who moves according
to~$\sigma$
until the robber is captured or visits~$r$ or~$s$. Assume by the symmetry of
argumentation,
that the robber visits~$r$. A second cop follows him to~$r$ and remains there
until the end
of the play. Then the first cop plays according to~$\sigma$ again. As~$r$ is
occupied by a cop,
the robber is either captured, or visits~$s$. Then the last cop follows him
to~$s$. Finally,
the first cop plays according to~$\sigma$ for the last time and the robber loses
the play.
\end{proof}

We show that the \cw of $\calG_n$ is bounded as well.
Graph~$\calG_n$ can be decomposed into~$n$ layers, the $i$-th layer is
induced by vertices $g_{i}, f_{i}, e_{i}, d_{i}, h_{i} ,k_{i}$, $t_{2i}$,
$t_{2i-1}$,
$a_{2i}$, and $a_{2i-1}$. We construct~$\calG_n$ inductively over~$i=1,\dots,n$
connecting the new layer to the previous ones.
Simultaneously, we connect~$r$, and~$s$ to the~$i$-th layer. Then vertex~$x$ is
connected to the~graph.

\begin{theorem}\label{thm_cw_switch_all}
For all~$n>0$, we have~$\cwm(\calG_n) \le 10$.
\end{theorem}
\begin{proof}
We consider graph~$\calG_n$ as consisting of layers~$\calL_i$,
for~$i\in\{1,\dots,n\}$ where each~$\calL_i$ is induced by 
vertices~$d_i$,~$e_i$,~$f_i$,~$h_i$,~$k_i$,~$g_i$,~$a_{2i}$,~$a_{2i-1}$,~$t_{2i}
$,
and~$t_{2i-1}$.
The layers are produced for~$i=1,2\dots,n$ by induction on~$i$ and connected to
the previous layers.
Level~$1$ is constructed
in the same way as further layers (up to vertex~$c$, which is easy to produce), 
so we do not describe the base case explicitly.
Assume, all layers from~$\calL_1$ to~$\calL_i$ are constructed with following
labelling,
which is an invariant that holds after a layer is constructed, see the first
picture
on Figure~\ref{fig_cw} (connections from~$t_i$ to~$r$ and~$s$ are not shown). 
\begin{itemize}
\item For~$j\in\{1,\dots,2i-1\}$, all~$t_j$, and, for~$j\in\{1,\dots,i\}$, all
$d_j$,
$e_j$, $h_j$, and $f_j$ have colour $\done$.

\item $t_{2i}$ has colour $T$.

\item  For~$j\in\{1,\dots,2i\}$, all $a_j$, have colour $A$.

\item For $j\in\{1,\dots,i\}$, all~$k_j$ have colour~$K$, and all~$g_j$, have
colour~$G$ .

\item $r$ has colour~$R$ and~$s$ has colour~$S$.
\end{itemize}
We construct layer~$i+1$ satisfying the invariant. First, create
vertex~$a_{2i+1}$
with colour~$A$ and vertex~$t_{2i+1}$ with colour~$T'$ and connect $A\to T'$.
Then
take the union of the previous layers and $\{a_{2i+1},t_{2i+1}\}$ and
connect~$T'\to T$,
$T'\to R$ and $T'\to S$. Relabel~$T'\to T$. Then repeat the same procedure
with~$a_{2i+2}$ and~$t_{2i+2}$
instead of~$a_{2i+1}$ and~$t_{2i+1}$, see the second picture
on Figure~\ref{fig_cw}.

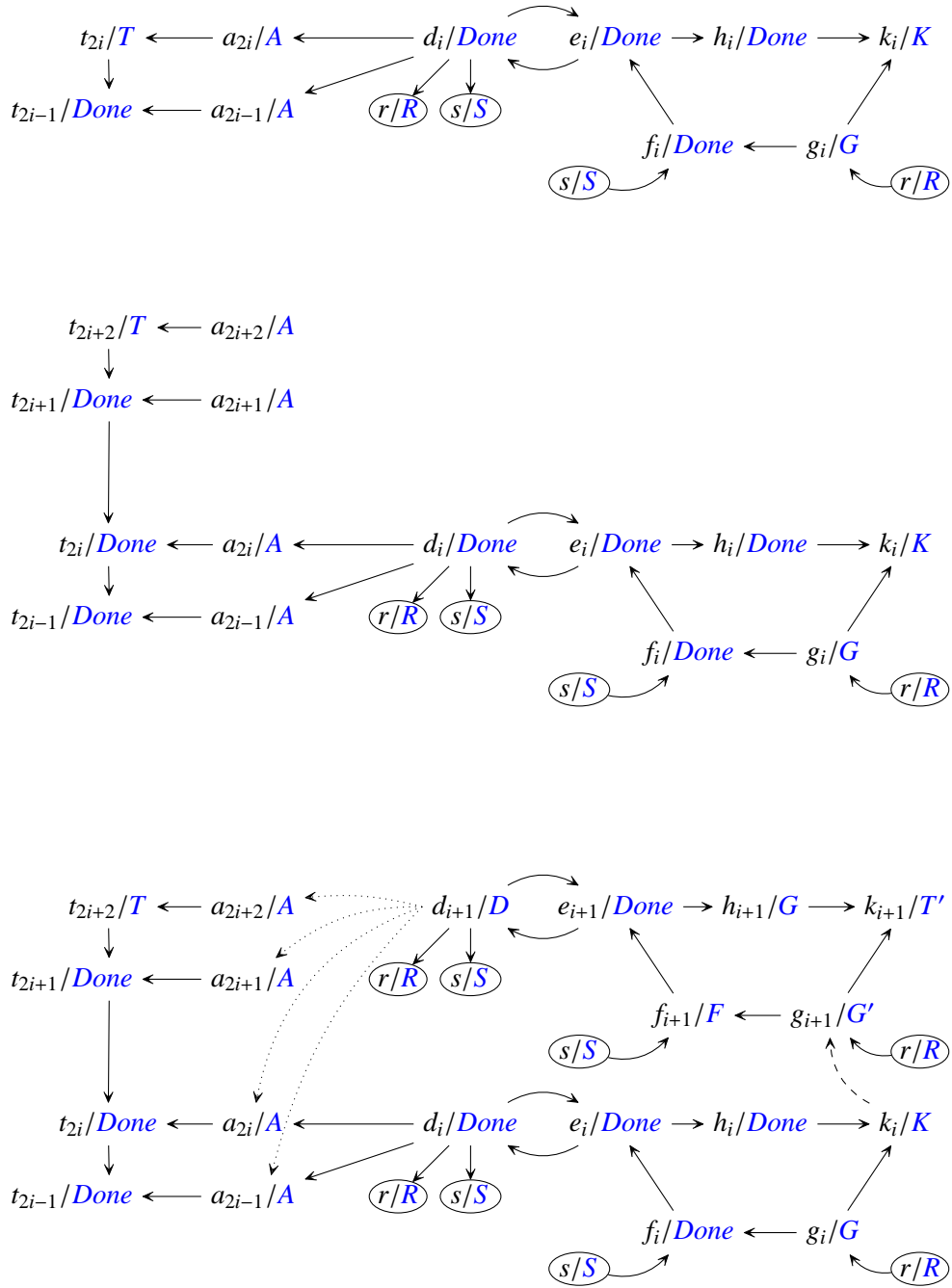
\begin{figure}
\begin{center}

\begin{tikzpicture}

\newcommand{\Done}{\color{blue}{Done}}
\newcommand{\T}{\color{blue}{T}}
\newcommand{\K}{\color{blue}{K}}
\newcommand{\G}{\color{blue}{G}}
\renewcommand{\R}{\color{blue}{R}}
\renewcommand{\S}{\color{blue}{S}}
\renewcommand{\A}{\color{blue}{A}}
\renewcommand{\F}{\color{blue}{F}}
\newcommand{\D}{\color{blue}{D}}


\node (d_i) at (0,0){$d_i/\Done$};
\node (e_i) at (2,0){$e_i/\Done$};
\node (h_i) at (4,0){$h_i/\Done$};
\node (k_i) at (6,0){$k_i/\K$};
\node (f_i) at (3,-1.5){$f_i/\Done$};
\node (g_i) at (5,-1.5){$g_i/\G$};
\node[draw,ellipse,inner sep=0] (r) at (6.2,-2){$r/\R$};
\node[draw,ellipse,inner sep=0] (s) at (1.5,-2){$s/\S$};

\draw[-slim,bend left] (d_i) to (e_i);
\draw[-slim,bend left] (e_i) to (d_i);
\draw[-slim] (e_i) to (h_i);
\draw[-slim] (h_i) to (k_i);
\draw[-slim] (g_i) to (k_i);
\draw[-slim] (g_i) to (f_i);
\draw[-slim] (f_i) to (e_i);
\draw[-slim,bend left] (r) to (g_i);
\draw[-slim,bend right] (s) to (f_i);

\node[draw,ellipse,inner sep=0] (r) at (-1,-1){$r/\R$};
\node[draw,ellipse,inner sep=0] (s) at (0,-1){$s/\S$};

\draw[-slim] (d_i) to (r);
\draw[-slim] (d_i) to (s);

\node (t_2i_1) at (-5.5,-1){$t_{2i-1}/\Done$};
\node (t_2i) at (-5,0){$t_{2i}/\T$};
\node (a_2i_1) at (-3,-1){$a_{2i-1}/\A$};
\node (a_2i) at (-3,0){$a_{2i}/\A$};

\draw[-slim] (a_2i_1) to (t_2i_1);
\draw[-slim] (a_2i) to (t_2i);
\draw[-slim] (d_i) to (a_2i);
\draw[-slim] (d_i) to (a_2i_1);
\draw[-slim] (t_2i) to (t_2i_1.30);


\node (d_i) at (0,-7){$d_i/\Done$};
\node (e_i) at (2,-7){$e_i/\Done$};
\node (h_i) at (4,-7){$h_i/\Done$};
\node (k_i) at (6,-7){$k_i/\K$};
\node (f_i) at (3,-8.5){$f_i/\Done$};
\node (g_i) at (5,-8.5){$g_i/\G$};
\node[draw,ellipse,inner sep=0] (r) at (6.2,-9){$r/\R$};
\node[draw,ellipse,inner sep=0] (s) at (1.5,-9){$s/\S$};

\draw[-slim,bend left] (d_i) to (e_i);
\draw[-slim,bend left] (e_i) to (d_i);
\draw[-slim] (e_i) to (h_i);
\draw[-slim] (h_i) to (k_i);
\draw[-slim] (g_i) to (k_i);
\draw[-slim] (g_i) to (f_i);
\draw[-slim] (f_i) to (e_i);
\draw[-slim,bend left] (r) to (g_i);
\draw[-slim,bend right] (s) to (f_i);

\node[draw,ellipse,inner sep=0] (r) at (-1,-8){$r/\R$};
\node[draw,ellipse,inner sep=0] (s) at (0,-8){$s/\S$};

\draw[-slim] (d_i) to (r);
\draw[-slim] (d_i) to (s);

\node (t_2i_1) at (-5.5,-8){$t_{2i-1}/\Done$};
\node (t_2i) at (-5,-7){$t_{2i}/\Done$};
\node (a_2i_1) at (-3,-8){$a_{2i-1}/\A$};
\node (a_2i) at (-3,-7){$a_{2i}/\A$};

\draw[-slim] (a_2i_1) to (t_2i_1);
\draw[-slim] (a_2i) to (t_2i);
\draw[-slim] (d_i) to (a_2i);
\draw[-slim] (d_i) to (a_2i_1);
\draw[-slim] (t_2i) to (t_2i_1.30);

\node (t_2i_p1) at (-5.5,-5){$t_{2i+1}/\Done$};
\node (t_2i_p2) at (-5,-4){$t_{2i+2}/\T$};
\node (a_2i_p1) at (-3,-5){$a_{2i+1}/\A$};
\node (a_2i_p2) at (-3,-4){$a_{2i+2}/\A$};

\draw[-slim] (a_2i_p1) to (t_2i_p1);
\draw[-slim] (a_2i_p2) to (t_2i_p2);
\draw[-slim] (t_2i_p2) to (t_2i_p1.30);
\draw[-slim] (t_2i_p1.-30) to (t_2i);


\node (d_i_p1) at (0,-12){$d_{i+1}/\D$};
\node (e_i) at (2,-12){$e_{i+1}/\Done$};
\node (h_i) at (4,-12){$h_{i+1}/\G$};
\node (k_i) at (6,-12){$k_{i+1}/\T'$};
\node (f_i) at (3,-13.5){$f_{i+1}/\F$};
\node (g_i_p1) at (5,-13.5){$g_{i+1}/\G'$};
\node[draw,ellipse,inner sep=0] (r) at (6.2,-14){$r/\R$};
\node[draw,ellipse,inner sep=0] (s) at (1.5,-14){$s/\S$};

\draw[-slim,bend left] (d_i_p1) to (e_i);
\draw[-slim,bend left] (e_i) to (d_i_p1);
\draw[-slim] (e_i) to (h_i);
\draw[-slim] (h_i) to (k_i);
\draw[-slim] (g_i_p1) to (k_i);
\draw[-slim] (g_i_p1) to (f_i);
\draw[-slim] (f_i) to (e_i);
\draw[-slim,bend left] (r) to (g_i_p1);
\draw[-slim,bend right] (s) to (f_i);

\node[draw,ellipse,inner sep=0] (r) at (-1,-13){$r/\R$};
\node[draw,ellipse,inner sep=0] (s) at (0,-13){$s/\S$};

\draw[-slim] (d_i_p1) to (r);
\draw[-slim] (d_i_p1) to (s);

\node (d_i) at (0,-15){$d_i/\Done$};
\node (e_i) at (2,-15){$e_i/\Done$};
\node (h_i) at (4,-15){$h_i/\Done$};
\node (k_i) at (6,-15){$k_i/\K$};
\node (f_i) at (3,-16.5){$f_i/\Done$};
\node (g_i) at (5,-16.5){$g_i/\G$};
\node[draw,ellipse,inner sep=0] (r) at (6.2,-17){$r/\R$};
\node[draw,ellipse,inner sep=0] (s) at (1.5,-17){$s/\S$};

\draw[-slim,bend left] (d_i) to (e_i);
\draw[-slim,bend left] (e_i) to (d_i);
\draw[-slim] (e_i) to (h_i);
\draw[-slim] (h_i) to (k_i);
\draw[-slim] (g_i) to (k_i);
\draw[-slim] (g_i) to (f_i);
\draw[-slim] (f_i) to (e_i);
\draw[-slim,bend left] (r) to (g_i);
\draw[-slim,bend right] (s) to (f_i);

\node[draw,ellipse,inner sep=0] (r) at (-1,-16){$r/\R$};
\node[draw,ellipse,inner sep=0] (s) at (0,-16){$s/\S$};

\draw[-slim] (d_i) to (r);
\draw[-slim] (d_i) to (s);

\node (t_2i_1) at (-5.5,-16){$t_{2i-1}/\Done$};
\node (t_2i) at (-5,-15){$t_{2i}/\Done$};
\node (a_2i_1) at (-3,-16){$a_{2i-1}/\A$};
\node (a_2i) at (-3,-15){$a_{2i}/\A$};

\draw[-slim] (a_2i_1) to (t_2i_1);
\draw[-slim] (a_2i) to (t_2i);
\draw[-slim] (d_i) to (a_2i);
\draw[-slim] (d_i) to (a_2i_1);
\draw[-slim] (t_2i) to (t_2i_1.30);

\node (t_2i_p1) at (-5.5,-13){$t_{2i+1}/\Done$};
\node (t_2i_p2) at (-5,-12){$t_{2i+2}/\T$};
\node (a_2i_p1) at (-3,-13){$a_{2i+1}/\A$};
\node (a_2i_p2) at (-3,-12){$a_{2i+2}/\A$};

\draw[-slim] (a_2i_p1) to (t_2i_p1);
\draw[-slim] (a_2i_p2) to (t_2i_p2);
\draw[-slim] (t_2i_p2) to (t_2i_p1.30);
\draw[-slim] (t_2i_p1.-30) to (t_2i);


\draw[-slim,bend left=30,dashed] (k_i.150) to (g_i_p1);
\draw[-slim,bend right=10,dotted] (d_i_p1.180) to (a_2i_p2);
\draw[-slim,bend right=20,dotted] (d_i_p1.180) to (a_2i_p1);
\draw[-slim,bend right=30,dotted] (d_i_p1.180) to (a_2i);
\draw[-slim,bend right=15,dotted] (d_i_p1.180) to (a_2i_1.50);

\end{tikzpicture}
\caption{Construction of $\calG_n$ with ten colours.}
\label{fig_cw}
\end{center}
\end{figure}

\begin{figure}[H]
\begin{center}

\begin{minipage}{0.6\linewidth}
\begin{tikzpicture}[scale=0.7]

\draw[use as bounding box,opacity=0,red] (-5,-3.5) rectangle (5,21);

\small

\newlength{\hoch}
\setlength{\hoch}{1.5cm}
\node[vertex] (k1) at (0,0) {$k_1$};
\node[vertex] (c10) at (-1,0.25) {$c_1^0$};
\node[vertex] (c11) at (1,0.25) {$c_1^1$};
\node[vertex] (A10) at ($ (c10) + (0,\hoch)  $) {$A_1^0$};
\node[vertex] (A11) at ($ (c11) + (0,\hoch)  $) {$A_1^1$};
\newlength{\cycle}
\setlength{\cycle}{0.8cm}
\node[vertex] (b110) at ($ (A10) + (-1,-\cycle)  $) {$b_{1,1}^0$};
\node[vertex] (b100) at ($ (A10) + (-1,\cycle)  $) {$b_{1,0}^0$};
\node[vertex] (b111) at ($ (A11) + (1,-\cycle)  $) {$b_{1,1}^1$};
\node[vertex] (b101) at ($ (A11) + (1,\cycle)  $) {$b_{1,0}^1$};

\node[dvertex] (k10) at ($ (b100) + (-1.5,0)  $) {$k_{[1;3]}$};
\node[dvertex] (k11) at ($ (b101) + (1.5,0)  $) {$k_{[1;3]}$};
\node[dvertex] (t10) at ($ (b110) + (-1.5,0)  $) {$t$};
\node[dvertex] (t11) at ($ (b111) + (1.5,0)  $) {$t$};

\node[vertex] (d10) at ($ (A10) + (0,2)  $) {$d_{1}^0$};
\node[vertex] (h10) at ($ (d10) + (0,1.2)  $) {$h_{1}^0$};
\node[vertex] (d11) at ($ (A11) + (0,2)  $) {$d_{1}^1$};
\node[vertex] (h11) at ($ (d11) + (0,1.2)  $) {$h_{1}^1$};

\node[dvertex] (s10) at ($ (d10) + (-1.3,0)  $) {$s$};
\node[dvertex] (s11) at ($ (d11) + (1.3,0)  $) {$s$};
\node[dvertex] (th1) at ($ (h10) + (-1.3,0)  $) {$t$};

\newlength{\x}
\setlength{\x}{0cm}
\newlength{\y}
\setlength{\y}{6.5cm}
\node[dvertex] (t2) at ($ (k1) + (\x,\y) + (2cm,-0.5cm) $) {$t$};
\node[vertex] (k2) at ($ (k1) + (\x,\y)  $) {$k_2$};
\node[vertex] (c20) at ($ (c10) + (\x,\y)  $) {$c_2^0$};
\node[vertex] (c21) at ($ (c11) + (\x,\y)  $) {$c_2^1$};
\node[vertex] (A20) at ($ (A10) + (\x,\y)  $) {$A_2^0$};
\node[vertex] (A21) at ($ (A11) + (\x,\y)  $) {$A_2^1$};

\node[vertex] (b210) at ($ (b110) + (\x,\y)  $) {$b_{2,1}^0$};
\node[vertex] (b200) at ($ (b100) + (\x,\y)  $) {$b_{2,0}^0$};
\node[vertex] (b211) at ($ (b111) + (\x,\y)  $) {$b_{2,1}^1$};
\node[vertex] (b201) at ($ (b101) + (\x,\y)  $) {$b_{2,0}^1$};

\node[dvertex] (k20) at ($ (k10) + (\x,\y)  $) {$k_{[1;3]}$};
\node[dvertex] (k21) at ($ (k11) + (\x,\y)  $) {$k_{[1;3]}$};
\node[dvertex] (t20) at ($ (t10) + (\x,\y)  $) {$t$};
\node[dvertex] (t21) at ($ (t11) + (\x,\y)  $) {$t$};

\node[vertex] (d20) at ($ (d10) + (\x,\y)  $) {$d_{2}^0$};
\node[vertex] (h20) at ($ (h10) + (\x,\y)  $) {$h_{2}^0$};
\node[vertex] (d21) at ($ (d11) + (\x,\y)  $) {$d_{2}^1$};
\node[vertex] (h21) at ($ (h11) + (\x,\y)  $) {$h_{2}^1$};

\node[dvertex] (s20) at ($ (d20) + (-1.3,0)  $) {$s$};
\node[dvertex] (s21) at ($ (d21) + (1.3,0)  $) {$s$};
\node[dvertex] (th2) at ($ (h20) + (-1.3,0)  $) {$t$};

\newlength{\xx}
\setlength{\xx}{0cm}
\newlength{\yy}
\setlength{\yy}{13cm}
\node[dvertex] (t3) at ($ (k1) + (\xx,14.3cm)  $) {$t$};
\node[vertex] (k3) at ($ (k1) + (\xx,\yy)  $) {$k_3$};
\node[vertex] (c30) at ($ (c10) + (\xx,\yy)  $) {$c_3^0$};
\node[vertex] (c31) at ($ (c11) + (\xx,\yy)  $) {$c_3^1$};
\node[vertex] (A30) at ($ (A10) + (\xx,\yy)  $) {$A_3^0$};
\node[vertex] (A31) at ($ (A11) + (\xx,\yy)  $) {$A_3^1$};

\node[vertex] (b310) at ($ (b110) + (\xx,\yy)  $) {$b_{3,1}^0$};
\node[vertex] (b300) at ($ (b100) + (\xx,\yy)  $) {$b_{3,0}^0$};
\node[vertex] (b311) at ($ (b111) + (\xx,\yy)  $) {$b_{3,1}^1$};
\node[vertex] (b301) at ($ (b101) + (\xx,\yy)  $) {$b_{3,0}^1$};

\node[dvertex] (k30) at ($ (k10) + (\xx,\yy)  $) {$k_{[1;3]}$};
\node[dvertex] (k31) at ($ (k11) + (\xx,\yy)  $) {$k_{[1;3]}$};
\node[dvertex] (t30) at ($ (t10) + (\xx,\yy)  $) {$t$};
\node[dvertex] (t31) at ($ (t11) + (\xx,\yy)  $) {$t$};

\node[vertex] (d30) at ($ (d10) + (\xx,\yy)  $) {$d_{3}^0$};
\node[vertex] (h30) at ($ (h10) + (\xx,\yy)  $) {$h_{3}^0$};
\node[vertex] (d31) at ($ (d11) + (\xx,\yy)  $) {$d_{3}^1$};
\node[vertex] (h31) at ($ (h11) + (\xx,\yy)  $) {$h_{3}^1$};

\node[dvertex] (s30) at ($ (d30) + (-1.3,0)  $) {$s$};
\node[dvertex] (s31) at ($ (d31) + (1.3,0)  $) {$s$};
\node[dvertex] (th3) at ($ (h30) + (-1.3,0)  $) {$t$};

\node[vertex] (s) at (0,-2.8) {$s$};
\node[dvertex] (ts) at (-1,-1.7) {$t$};
\node[dvertex] (ks) at (1,-1.7) {$k_{1;3}$};
\draw[-slim] (s) to (ts);
\draw[-slim] (s) to (ks);

\node[vertex] (t) at (-1,19.5) {$t$};
\node[vertex] (k4) at (1,19.5) {$k_4$};
\draw[-slim] (k4) to (t);
\draw[-slim] (h31) to (k4);
\draw[slim loop] (t) to (t);

\draw[-slim] (k1) to (ts);
\draw[-slim] (k1) to (c10);
\draw[-slim] (k1) to (c11);
\draw[-slim] (c10) to (A10);
\draw[-slim] (c11) to (A11);

\draw[-slim] (A10) to (d10);
\draw[-slim] (A11) to (d11);
\draw[-slim] (d10) to (h10);
\draw[-slim] (d11) to (h11);

\draw[-slim,bend left] (b110) to (A10);
\draw[-slim,bend left] (b100) to (A10);
\draw[-slim,bend left] (A10) to (b110);
\draw[-slim,bend left] (A10) to (b100);

\draw[-slim,bend left] (b111) to (A11);
\draw[-slim,bend left] (b101) to (A11);
\draw[-slim,bend left] (A11) to (b111);
\draw[-slim,bend left] (A11) to (b101);

\draw[-slim] (b110) to (k10);
\draw[-slim] (b100) to (k10);
\draw[-slim] (b110) to (t10);
\draw[-slim] (b100) to (t10);

\draw[-slim] (b111) to (k11);
\draw[-slim] (b101) to (k11);
\draw[-slim] (b111) to (t11);
\draw[-slim] (b101) to (t11);

\draw[-slim] (d10) to (s10);
\draw[-slim] (d11) to (s11);
\draw[-slim] (h10) to (th1);

\draw[-slim] (k2) to (t2);
\draw[-slim] (k2) to (c20);
\draw[-slim] (k2) to (c21);
\draw[-slim] (c20) to (A20);
\draw[-slim] (c21) to (A21);

\draw[-slim] (A20) to (d20);
\draw[-slim] (A21) to (d21);
\draw[-slim] (d20) to (h20);
\draw[-slim] (d21) to (h21);

\draw[-slim,bend left] (b210) to (A20);
\draw[-slim,bend left] (b200) to (A20);
\draw[-slim,bend left] (A20) to (b210);
\draw[-slim,bend left] (A20) to (b200);

\draw[-slim,bend left] (b211) to (A21);
\draw[-slim,bend left] (b201) to (A21);
\draw[-slim,bend left] (A21) to (b211);
\draw[-slim,bend left] (A21) to (b201);

\draw[-slim] (b210) to (k20);
\draw[-slim] (b200) to (k20);
\draw[-slim] (b210) to (t20);
\draw[-slim] (b200) to (t20);

\draw[-slim] (b211) to (k21);
\draw[-slim] (b201) to (k21);
\draw[-slim] (b211) to (t21);
\draw[-slim] (b201) to (t21);

\draw[-slim] (d20) to (s20);
\draw[-slim] (d21) to (s21);
\draw[-slim] (h20) to (th2);

\draw[-slim] (k3) to (t3);
\draw[-slim] (k3) to (c30);
\draw[-slim] (k3) to (c31);
\draw[-slim] (c30) to (A30);
\draw[-slim] (c31) to (A31);

\draw[-slim] (A30) to (d30);
\draw[-slim] (A31) to (d31);
\draw[-slim] (d30) to (h30);
\draw[-slim] (d31) to (h31);

\draw[-slim,bend left] (b310) to (A30);
\draw[-slim,bend left] (b300) to (A30);
\draw[-slim,bend left] (A30) to (b310);
\draw[-slim,bend left] (A30) to (b300);

\draw[-slim,bend left] (b311) to (A31);
\draw[-slim,bend left] (b301) to (A31);
\draw[-slim,bend left] (A31) to (b311);
\draw[-slim,bend left] (A31) to (b301);

\draw[-slim] (b310) to (k30);
\draw[-slim] (b300) to (k30);
\draw[-slim] (b310) to (t30);
\draw[-slim] (b300) to (t30);

\draw[-slim] (b311) to (k31);
\draw[-slim] (b301) to (k31);
\draw[-slim] (b311) to (t31);
\draw[-slim] (b301) to (t31);

\draw[-slim] (d30) to (s30);
\draw[-slim] (d31) to (s31);
\draw[-slim] (h30) to (th3);

\draw[-slim] (h11) to ($ (h11) + (0,1) $) to (k2);
\draw[-slim] (h21) to ($ (h21) + (0,1) $) to (k3);
\draw[-slim] (h10) to ($ (h10) + (0,0.5) $) to
($ (h10) + (-3.5,0.5) $) to
($ (h10) + (-3.5,7) $) to (k3);

\draw[-slim] (k1) to ($ (k1) + (-0.3,0.6) $) to
($ (k1) + (-0.3,5.7) $) to
($ (k1) + (-4.1,5.7) $) to
($ (k1) + (-4.1,11.5) $) to (k3);
\draw[-slim] ($ (k1) + (-0.3,5.7) $) to (k2);

\draw[-slim] (k2) to ($ (k1) + (0.3,5.7) $) to
($ (k1) + (0.3,0.6) $) to (k1);
\draw[-slim] (k2) to ($ (k2) + (-0.3,0.6) $) to
($ (k2) + (-0.3,5.7) $) to (k3);

\draw[-slim] (k3) to ($ (k2) + (0.3,5.7) $) to
($ (k2) + (0.3,0.6) $) to (k2);
\draw[-slim] (k3) to ($ (k1) + (4.1,11.5) $) to
($ (k1) + (4.1,-0.8) $) to ($ (k1) + (0,-0.8) $) to (k1);

\end{tikzpicture}
\end{minipage}
\begin{minipage}{0.25\linewidth}
\begin{center}
   
\begin{tabular}{|c|c|c|c|}
\hline
Node 	& Successors \\
\hline
$d_{i}^{j}$ & $h_{i}^{j},s$  \\
$A_{i}^{j}$ & $d_{i}^{j},b_{i,0}^{j},b_{i,1}^{j}$ \\
$b_{i,*}^{j}$ & $t,A_{i}^{j},k_{[1;n]}$ \\
$t$ & $t$ \\
$s$ & $t,k_{[1;n]}$ \\
$k_{n+1}$ & $t$\\
$k_i$ & $c_{i}^{0}c_{i}^{1},t,k_{[1;n]}$\\
$h_{i}^{0}$ & $t,k_{[i+2;n]}$\\
$h_{i}^{1}$ & $k_{i+1}$\\
$c_{i}^{j}$ & $A_{i}^{j}$\\
\hline
\end{tabular}
\end{center}

\end{minipage}
\caption{The graph~$\calZ_3$ and the edge relation of~$\calZ_n$ for the
counterexample to Zadeh's least-entered rule.}
\label{fig_Zadeh}
\end{center}
\end{figure}
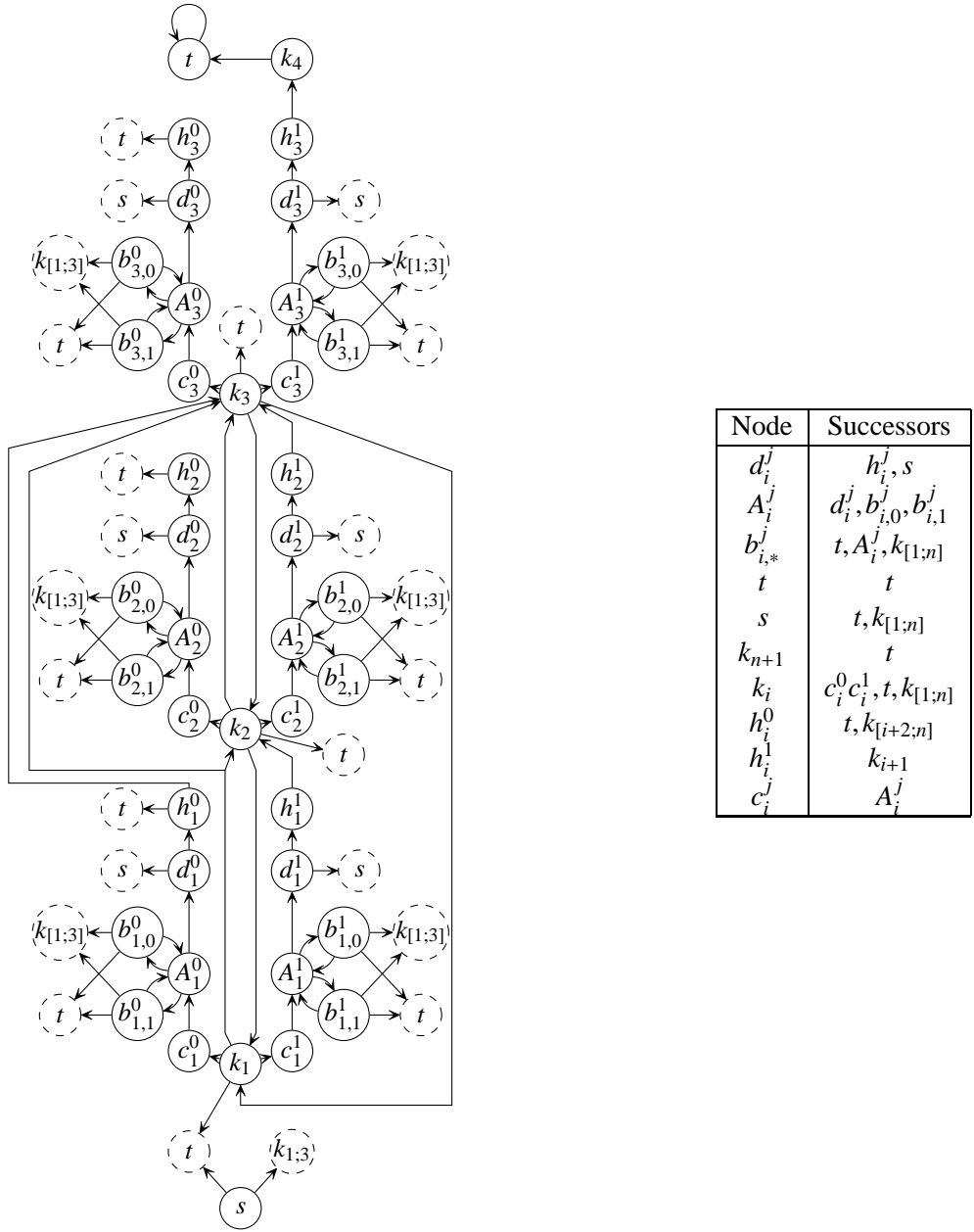

Now we construct the subgraph~$\calC_{i+1}$ induced 
by $d_{i+1}$, $e_{i+1}$, $h_{i+1}$, $g_{i+1}$, and~$f_{i+1}$ using 
colours $D$, $\done$, $G$, $T'$, $G'$, and~$F$. Note that colours~$\done$
and~$T'$
are reused.
Produce $d_{i+1}$, $e_{i+1}$, $f_{i+1}$, $g_{i+1}$, $k_{i+1}$, and~$h_{i+1}$
with
labels $D$, $\done$, $F$, $G'$, $T'$, and~$G$, respectively, and connect them as
needed,
also to~$r$ and to~$s$ (see the third picture on Figure~\ref{fig_cw}).

Relabel~$G\to \done$. Build the disjoint union of~$\calC_{i+1}$ and the already
constructed graph.
Connect~$K\to G'$ (which connects all~$k_j$, for $j<i+1$ to~$g_{i+1}$; dashed
line in the figure),
and relabel~$T'\to k$ and~$G'\to G$. Connect~$D\to A$ (which connects~$d_{i+1}$
to all~$a_j$, 
for $j\le 2i$, dotted lines in the figure) and relabel~$D\to\done$.
This finishes the construction of layer~$L_{i+1}$. Note that the properties from
the invariant hold
for~$L_{i+1}$. Finally, produce vertex~$x$ with colour $T'$ and connect
all~$k_i$ to~$x$ and~$x$ to itself. 
It remains to count the colours. We used $\done$, $T$, $T'$, $A$, $F$, $G$,
$G'$, $K$, $R$, and $S$,
which makes ten colours.

\end{proof}

\subsection{Zadeh's least-entered rule}

As a second example we discuss the counterexample of Friedmann
against Zadeh's least-entered rule. The underlying game graphs
are denoted $\calZ_n$. As $\calG_n$
the graph $\calZ_n$ can be decomposed into~$n$ layers, see
Figure~\ref{fig_Zadeh} for graph~$\calZ_3$ and a list of edges of~$\calZ_n$.
Vertices~$k_i$,
~$c_i^0$,~$A_i^0$,~$b_{i,1}^0$,~$b_{i,0}^0$,~$d_i^0$,~$h_i^0$,
~$c_i^1$,~$A_i^1$,~$b_{i,1}^1$,~$b_{i,0}^1$,~$d_i^1$, and~$h_i^1$ induce the
$i$-th layer.
The subgraphs induced
by~$c_i^0$,~$A_i^0$,~$b_{i,1}^0$,~$b_{i,0}^0$,~$d_i^0$,~$h_i^0$ and
~$c_i^1$,~$A_i^1$,~$b_{i,1}^1$,~$b_{i,0}^1$,~$d_i^1$,~$h_i^1$ are isomorphic to
each other.

A run of the strategy improvement alogrithm on~$\calZ_n$ simulates an $n$-bit
counter with values from~0 to~$2^n-1$. The difference to the switch-all rule is
that the least-entered rule chooses an improving edge that has been switched
least often.
Because the lower bits of an $n$-bit counter are switched more often,
the higher bits would be switched before they should in order to catch
up with the lower bits. This means that the $n$-bit counter would not go through
all the steps from~0 to~$2^n-1$. Friedmann solved this problem by representing
each
bit~$i$ by two bits,~$i_0$ and~$i_1$. The associated structures in~$\calZ_n$ are
the
gadgets induced by $\{A_i^0,b_{i,1}^0,b_{i,0}^0\}$ and
$\{A_i^1,b_{i,1}^1,b_{i,0}^1\}$
respectively. The bit~$i_j$ is considered to be set, if the current
Player~0 strategy chooses both edges~$(b_{i,0}^j, A_i^j)$
and~$(b_{i,1}^j, A_i^j)$, and unset otherwise. In a run of the algorithm
only one of the bits~$i_0$ and~$i_1$ is active and able to
effect the rest of the counter at
the time. The inactive bit can, in the meantime, switch back and forth from~0
to~1 in order to catch up with the rest of the counter without having an effect
on it.

The counterexample contains a vertex~$k_i$ in each layer~$i$ such that
all~$k_i$ induce an $n$-clique in the graph. This makes all values of measures
that
describe cyclicity (\ie, \tw, \dpw, \dagw, and \kw) unbounded on the class of
the counterexample graphs, but \cw of the graphs is still small. 

\begin{theorem}\label{thm_cw_zadeh}
For all~$n>0$, we have~$\cwm(\calZ_n) \le 9$.
\end{theorem}
\begin{proof}
The proof is very similar to the proof of Theorem~\ref{thm_cw_switch_all}.
We regard the graphs~$\calZ_n$ as consisting of layers~$\calL_i$ that are
induced by vertices~$k_i$,
$c_i^0$, $A_i^0$, $b_{i,0}^j$, $b_{i,1}^j$, $d_i^j$, and~$h_i^j$,
for $i\in\{1,\dots,n\}$ and~$j\in\{0,1\}$.
The layers are constructed for~$i=1,2,\dots,n$ by induction on~$i$
and connected to the previous layers. 
In the induction step, we build a new layer and connect it to the previous ones.
Finally, we add
the vertex~$s$ and the top layer, that consists of~$t$ and~$k_{n+1}$, and
establish
the connections to the other~$n$ layers.

As in the proof of Theorem~\ref{thm_cw_switch_all}, layer~$\calL_1$ is 
constructed in the same
way as further layers. Assume, layers from~$\calL_1$ to~$\calL_i$ have
been constructed with the following labelling, which is an invariant that
holds after a new layer is constructed.
\begin{itemize}
\item For $j\in\{1,\dots,i\}$ and $s,s'\in\{1,2\}$, all~$k_j$ have colour~$K$,
all~$A_j^s$ and~$c_j^s$ have colour~$Done$, all~$d_j^s$ have colour~$D$, and
all~$b_{i,s'}^s$ have colour~$B$.

\item For $j\in\{1,\dots,i-1\}$, all~$h_j^0$ have colour~$H$ and all~$h_j^1$
have
colour~$Done$.

\item $h_i^0$ has colour~$H_l$ and~$h_i^1$ has colour~$H_r$.

\end{itemize}
We construct the layer~$\calL_{i+1}$ and connect it to the previous
layers such that at the end of that process the invariant is true.
First, produce
the vertex~$k_{i+1}$ with colour~$K'$. Connect the vertices
~$h_1^0,\dots,h_{i-1}^0$ and the vertex~$h_i^1$ to~$k_{i+1}$ by connecting
$H$-ports and $H_r$-ports to $K'$-ports, and relabelling 
$H_l \rightarrow H$ and $H_r \rightarrow Done$. Extend the clique consisting
of the vertices~$k_1,\dots,k_i$ by connecting
$K \rightarrow K'$ and $K' \rightarrow K$. Thus the conncections between
~$\calL_{i+1}$ and the previous layers have been established.

Next, we construct the rest of~$\calL_{i+1}$ using colours~$C$, $Done$, $B$,
$D$, $H_l$ and~$H_r$. Create vertices~$c_{i+1}^0$, $A_{i+1}^0$,
$b_{i+1,1}^0$, $b_{i+1,0}^0$, $d_{i+1}^0$ and~$h_{i+1}^0$ labelled with
colours~$C$, $Done$, $B$, $B$, $D$ and~$H_l$, respectively, and connect them as
needed. Repeat this procedure for vertices~$c_{i+1}^1$, $A_{i+1}^1$,
$b_{i+1,1}^1$, $b_{i+1,0}^1$, $d_{i+1}^1$ and~$h_{i+1}^1$ with the difference
that~$h_{i+1}^1$ obtains colour~$H_r$. Build the disjoint union of
these two subgraphs and the already constructed graph. Connect~$k_{i+1}$
to~$c_{i+1}^0$ and~$c_{i+1}^1$ by $K'\rightarrow C$. Relabel $K'\rightarrow K$
and $C\rightarrow Done$. This finishes the construction of~$\calL_{i+1}$. Note
that the invariant for the vertex labels is satisfied.

After all~$n$ layers have been built, relabel $H_l \rightarrow H$ and
create vertex~$s$ with colour~$K'$ (which is reused). Connect $K'\rightarrow K$
and $D\rightarrow K'$. Relabel $D\rightarrow Done$. It remains to add
vertices~$k_{n+1}$ and~$t$ to the graph.
Create~$k_{n+1}$ and~$t$ with colours~$C$ and~$D$.
Connect $H_r\rightarrow C$, $C\rightarrow D$, $D\rightarrow D$,
$K\rightarrow D$, $B\rightarrow K$, $B\rightarrow D$
and~$K'\rightarrow D$. This produces the top layer induced by~$k_{i+1}$
and~$t$ and establishes the edges between the first~$n$ layers and the
vertex~$s$ and the top layer. Note that we reused the colours~$D$,~$C$ and~$K'$.
It remains to count the colours. We used the nine colours~$Done$,
~$K$,~$K'$,~$C$,~$B$,~$D$,~$H_r$,~$H_l$ and~$H$. Hence, our claim holds.
\end{proof}

\subsection{Other rules}\label{subsec_other_rules}

The graph complexity of the other counterexamples constructed by Friedmann is
similar to the one
of the graphs~$\calG_n$. All of them can be decomposed into layers which are
connected in a simple
way, and they can be treated by the same techniques as in 
Theorems~\ref{thm_dpw_bounded},~\ref{thm_ent_bounded}
and~\ref{thm_cw_switch_all}. The results
are shown in Table~\ref{table_results}. For the 
case of snare memorisation we did not find any bound for the \cw, but all the
other measures (except \tw) have very small values on those graphs.

\begin{table}[H]
\begin{tabular}{l | c c c c c c }
Rule / Measure & tree & directed path & DAG & Kelly & \Ent & Clique \\
\hline
switch-all					& $\infty$ & 3 & 4 & 4 & 3 & 10
\\
switch-best					& $\infty$ & 3 & 4 & 4 & 3 & 18
\\
random-edge					& 8 & 3 & 4 & 4 & 3 & 12 \\
random-facet				& 3 & 1 & 2 & 2 & 1 & 6 \\
least-entered				& $\infty$ & $\infty$ & $\infty$ &
$\infty$ & $\infty$ & 9 \\
least-considered		& 7 & 3 & 4 & 4 & 4 & 7 \\
snare memory				& $\infty$ & 3 & 4 & 4 & 4 & ? \\
\end{tabular}
\caption{Upper bounds in different measures for the counterexample graph
classes.}
\label{table_results}
\end{table}

\bibliographystyle{eptcs}
\bibliography{simple_counterexamples_short}

\end{document}